\documentclass[fleqn,12pt,twoside]{article}

\usepackage{graphicx}
\input{extern/layout_commands}
\usepackage{amsmath}
\usepackage{amssymb}

\newcounter{mynumcounter}
	       {%
		 \begin{list}{(\roman{mynumcounter})\hspace*{\fill}}%
		   {
		     \setlength{\topsep}{0cm}
		     \setlength{\partopsep}{0cm}
		     \setlength{\itemsep}{0ex}
		     \setlength{\parsep}{0cm}
		     \setlength{\leftmargin}{0cm}
		     \setlength{\itemindent}{8mm}		   
		     \setlength{\labelsep}{3mm}
		     \setlength{\labelwidth}{5mm}
		     \usecounter{mynumcounter}
		   }%
	       }
	       {\end{list}}

\newcommand{\eps}{\varepsilon}
\newcommand{\del}{\partial}

\newcommand{\dd}[2]{\frac{\del #1}{\del #2}}
\newcommand{\ddeval}[3]{\left.\dd{#1}{#2}\right|_{#3}}

\newcommand{\tr}{\operatorname{tr}}
\renewcommand{\div}{\operatorname{div}}

\newcommand{\IR}{\mathbf{R}}

\newcommand{\CH}{\mathcal{H}}

\newcommand{\CT}{\mathcal{T}}
\newcommand{\CU}{\mathcal{U}}

\newcommand{\al}{\alpha}

\newcommand{\rmd}{\mathrm{d}}

\newcommand{\dist}{\mathrm{dist}}

\newcommand{\dmu}{\,\rmd\mu}

\newcommand{\lap}{\Delta}
\newcommand{\Laplacian}[1]{\smash{\sideset{^{#1}}{}{\mathop\lap\nolimits}}}

\newcommand{\lapSig}{\Laplacian{\Sigma}}

\newcommand{\Connection}[1]{\smash{\sideset{^{#1}}{}{\mathop\nabla\nolimits}}}

\newcommand{\nabSig}{\!\Connection{\Sigma}}

\newcommand{\nabL}{\!\Connection{L}}

\newcommand{\Divergence}[1]{\smash{\sideset{^{#1}}{}{\mathop\mathrm{div}\nolimits}}}

\newcommand{\divSig}{\!\Divergence{\Sigma}}

\newcommand{\Ricci}[1]{\smash{\sideset{^{#1}}{}{\mathop\mathrm{Rc}\nolimits}}}

\newcommand{\RicL}{\!\Ricci{L}}

\newcommand{\Scalarcurv}[1]{\smash{\sideset{^{#1}}{}{\mathop\mathrm{Sc}\nolimits}}}

\newcommand{\ScalM}{\!\Scalarcurv{M}}
\newcommand{\ScalSig}{\!\Scalarcurv{\Sigma}}

\newcommand{\Acircbar}%
{\hspace*{4pt}\raisebox{8.5pt}{\makebox[-4pt][l]{$\scriptstyle\circ$}}\bar A}
\newcommand{\Kcircbar}%
{\hspace*{4pt}\raisebox{8.5pt}{\makebox[-4pt][l]{$\scriptstyle\circ$}}\bar K}

\input{extern/draft_commands}

\title{\titlefamily\Huge 
The time evolution of marginally trapped surfaces
}
\author{
  \authname{Lars Andersson$^{\star,\dagger}$}
  \and
  \authname{Marc Mars$^{\ddagger}$}
  \and
  \authname{Jan Metzger$^{\star}$}
  \and
  \authname{Walter Simon$^{\ddagger}$}
  \\[2mm]
  \authaddress{    
    $^\star$\ 
    Albert-Einstein-Institut,
    Am M\"uhlenberg 1,
    D-14476 Potsdam,
    Germany
  }
  \authaddress{
    $^\dagger$\ 
    Department of Mathematics,
    University of Miami,
    Coral Gables, FL 33124,
    USA
  }
  \authaddress{
    $^\ddagger$\ 
    Facultad de Ciencias,
    Universidad de Salamanca,\\
    Plaza de la Merced s/n,
    E-37008 Salamanca, Spain
  }
}
\date{}

\begin{document}
\hyphenation{}
\pagestyle{footnumber}
\maketitle
\thispagestyle{footnumber}
\begin{abst}%
  \blfootnote{Email:
    \begin{tabular}[t]{llll}
      LA & lars.andersson@aei.mpg.de, &
      MM & marc@usal.es,\\
      JM & jan.metzger@aei.mpg.de, &
      WS & walter@usal.es.
    \end{tabular}}
  In previous work we have shown the existence of a dynamical horizon
  or marginally trapped tube (MOTT) containing a given strictly stable
  marginally outer trapped surface (MOTS). In this paper we show some
  results on the global behavior of MOTTs assuming the null energy
  condition. In particular we show that MOTSs persist in the sense
  that every Cauchy surface in the future of a given Cauchy surface
  containing a MOTS also must contain a MOTS. We describe a situation
  where the evolving outermost MOTS must jump during the coalescence
  of two seperate MOTSs.  We furthermore characterize the behavior of
  MOTSs in the case that the principal
  eigenvalue vanishes under a genericity assumption. This leads to a
  regularity result for the tube of outermost MOTSs under the
  genericity assumption. This tube is then smooth up to finitely many
  jump times. Finally we
  discuss the relation of MOTSs to singularities of a space-time.
\end{abst}
\section{Introduction}
In previous work \cite{Andersson-Metzger:2005,Andersson-Metzger:2007},
we considered marginally trapped surfaces, or more specifically,
marginally outer trapped surfaces (MOTS). These were studied in the
context of initial data sets without regarding their time
dependence. In this note we shall consider the behavior of outermost
MOTSs in the context of Cauchy slicings. The main result in this
respect is that when a MOTS exists initially, it persists, provided
the developing spacetime satisfies the null energy
condition. Moreover, the domain bounded by these outermost MOTSs, in a
given Cauchy slice, contains the intersection of the causal future of
the initial one with the Cauchy slice. This is a generalization of the
fact, that in the smooth case a strictly stable MOTS gives rise to a
spacelike tube foliated by MOTSs in its vicinity
\cite{Andersson-Mars-Simon:2005,andersson-mars-simon:2007pub}. This is
the content of section~\ref{sec:evol-outerm-mots}. Then we discuss
some questions of regularity of the so defined family of MOTSs. To get
started, we show in section~\ref{sec:coalescence-mots} that sometimes
the outermost MOTS must jump. We consider the case where two bodies
with separate MOTSs surrounding them come close enough together. In
such a scenario the outermost MOTS jumps before the individual MOTSs
make contact. In section~\ref{sec:jump-targ-outerm} we analyze the
targets of such jumps under a genericity condition. We find that the
target of such a jump generically is part of a marginally outer
trapped tube which is tangent to the time slice at the jump time and
lies to the future of that slice. In section~\ref{sec:regularity} we
look at some global regularity properties of the family of outermost
MOTS. In particular we show that generically jumps are the only
singularities that can happen, and that they are discrete. Finally, in
section~\ref{sec:relat-outer-trapp} we conclude with a version of the
well-known singularity theorems which works for outer trapped
surfaces. The ideas for its proof are all present in the literature,
although we were not able to find the precise statement of the given
theorem. The closest references are probably in \cite{gannon:1976}
and \cite{totschnig:1994}.


%
\section{Preliminaries}
\label{sec:preliminaries}
We consider data sets for the Einstein equations. These are triples
$(M,g,K)$ where $M$ is a compact 3-manifold with boundary, $g$ a
Riemannian metric and $K$ a symmetric 2-tensor on $M$. We assume that
$\del M$ has two disconnected parts $\del M = \del^- M \cup
\del^+M$. We equip the \emph{inner} boundary $\del^-M$ with the normal
vector pointing into $M$ and the \emph{outer} boundary $\del^+M$ with
the normal vector pointing out of $M$.

Assume that $\Sigma\subset M$ is a surface in the interior of $M$ that
encloses a region $\Omega$ together with the outer boundary $\del^+
M$, that is $\del \Omega = \Sigma \cup \del^+ M$. If $\Sigma$ is
embedded, then this is equivalent to the condition that $\Sigma$ be
homologous to $\del^+M$ and we choose the \emph{outer normal} on
$\Sigma$ as the vector field pointing into $\Omega$, that is in
direction of the outer boundary. This vector is denoted by $\nu$
subsequently.

For a surface $\Sigma$ homologous to $\del^+M$, we define the
\emph{outgoing null expansion} as
\begin{equation*}
  \theta^+[\Sigma] = P + H
\end{equation*}
where $P = \tr K - K(\nu,\nu)$ and $H$ is the mean curvature of
$\Sigma$ with respect to the outer normal as defined above.

We say that a surface $\Sigma$ which is homologous to $\del^+ M$ is a
marginally outer trapped surface (MOTS) if
\begin{equation*}
  \theta^+[\Sigma] = 0.
\end{equation*}

We define an \emph{outermost MOTS} in $M$ to be an embedded MOTS
$\Sigma$ homologous to $\del^+M$ which bounds a region $\Omega$
together with $\del^+ M$, such that for any other such MOTS $\Sigma'$,
homologous to $\del^+M$, bounding $\Omega'$ together with $\del^+ M$,
it holds that if $\Omega' \subset \Omega$, then $\Sigma'=\Sigma$. This
is the global notion of being outermost also used
in~\cite{Andersson-Metzger:2007}.

Given a MOTS $\Sigma$ in $M$ we linearize the operator
$\theta^+$ near $\Sigma$ in the following way. Given a function $f$ on
$\Sigma$, define the surface $\Sigma_f$ as the image of the
parametrization
\begin{equation*}
  G^M_f :\Sigma \to M : x \mapsto \exp_x (f(x)\nu),
\end{equation*}
where $\nu$ is the outer normal to $\Sigma$ and $\exp$ the exponential
map of $M$. It is clear that if $f$ is smooth and $\eps$ is small
enough, $\Sigma_{\eps f}$ is a smooth embedded surface if $\Sigma$
is. The linearization of the operator $\theta^+$ at $f=0$ is then
given by the following linear, elliptic second order differential
operator
\begin{equation*}
  \begin{split}
    &\ddeval{}{\eps}{\eps=0} G^M_f\circ\theta^+[\Sigma_{\eps f}]
    = L_M f
    \\
    &\quad
    =  -\lapSig f + 2 S(\nabSig f) + f\big(\divSig S
    -\tfrac{1}{2}|\chi^+|^2 - |S|^2 + \tfrac{1}{2}\ScalSig - \mu - J(\nu) \big).
  \end{split}
\end{equation*}
Here $\lapSig$, $\nabSig$ and $\divSig$ are the Laplace-Beltrami
operator, the tangential gradient and the divergence along $\Sigma$,
$\chi^+ = A + K^{\Sigma}$ with $A$ the second fundamental form of
$\Sigma$ in $M$ and $K^\Sigma$ the tangential projection of $K$ to
$\Sigma$. Furthermore $S(\cdot) = K(\nu,\cdot)^T$, where $(\cdot)^T$
denotes orthogonal projection to $T\Sigma$. $\ScalSig$ is the scalar
curvature of $\Sigma$, $\mu = \tfrac{1}{2}\big( \ScalM - |K|^2 + (\tr
K)^2 \big)$, and $J = \div K - d(\tr K)$.

The operator $L_M$ has a unique eigenvalue $\lambda$ which minimizes
the real part in the spectrum of $L_M$. $\lambda$ is real, the
corresponding eigenspace is one-dimensional and the non-zero functions
in this eigenspace have a sign. $\lambda$ is called the principal
eigenvalue of $L_M$. We say that a MOTS $\Sigma$ is \emph{stable} if
$\lambda \geq 0$ and \emph{strictly stable} if $\lambda>0$. When
referring to the principal eigenvalue of a MOTS subsequently we always
mean the principal eigenvalue of $L_M$ on that MOTS. Further details
can be found
in~\cite{Andersson-Mars-Simon:2005,Andersson-Metzger:2005,andersson-mars-simon:2007pub}.

From \cite[Section 7]{Andersson-Metzger:2007}, we recall the following
notion. For embedded surfaces $\Sigma$ homologous to $\del ^+ M$
bounding a region $\Omega$ together with $\del^+ M$, we say that the
interior set $U := M\setminus \Omega$ is called \emph{weakly outer
  trapped set} if $\theta^+[\Sigma] \leq 0$. The \emph{weakly outer
  trapped region} $\CT$ of $M$ is the union of all weakly outer
trapped sets in $M$:
\begin{equation*}
  \CT := \bigcup\, \big\{\Omega : \Omega\ \text{is weakly outer trapped}\big\}.
\end{equation*}
For brevity, we will call $\CT$ the \emph{trapped region}. Under the
above assumptions, if $\del^- M$ is non-empty and has
$\theta^+[\del^-M] < 0$, the trapped region $\CT$ will also be
non-empty and include a neighborhood of $\del^- M$. Thus it makes
sense to define the outer boundary of $\CT$ as
\begin{equation*}
  \del^+ \CT = \del \CT \setminus \del^- M.
\end{equation*}
The following theorem was proved in \cite[Theorem
7.3]{Andersson-Metzger:2007}, see also \cite{andersson:2009ere}.
\begin{theorem}
  \label{thm:trapped-region}
  Let $(M,g,K)$ be as described above with $\theta^+[\del^-M] <0$ and
  $\theta^+[\del^+ M] >0$. Then the outer boundary $\del^+\CT$ of the
  trapped region is a smooth, stable, embedded MOTS.
\end{theorem}
Further properties of outermost MOTS derived in
\cite{Andersson-Metzger:2007} include the following estimates.
\begin{theorem}
  \label{thm:estimates}
  Assume that $(M,g,K)$ has $\theta^+[\del^-M] <0$ if $\del^- M$ is
  non-empty and $\theta^+[\del^+ M]>0$. Then there exist constants $C$
  and $\delta>0$ depending only on the geometry of $(M,g,K)$ with the
  following property.

  If $\Sigma$ is an outermost MOTS homologous to $\del^+M$ in
  $(M,g,K)$ then
  \begin{equation*}
    |A| \leq C
    \qquad\text{and}\qquad
    i^+(\Sigma) \geq \delta.
  \end{equation*}  
  Here $|A|$ is the norm of the second fundamental form of $\Sigma$ in
  $M$ and $2i^+(\Sigma)$ is the minimum length that a geodesic
  starting on $\Sigma$ in direction of the outer normal must travel
  before it can meet $\Sigma$ a second time.
\end{theorem}
Finally, we introduce some notation. Assume that $(L,h)$ is a
Lorentzian spacetime manifold with boundary, foliated by spacelike
slices
\begin{equation*}
  L = M\times I
\end{equation*}
where $ I\subset \IR$ is some interval and $M$ a three dimensional
manifold as above.  We choose the time orientation on $L$ so that
$t\in I$ increases to the future. We denote $M_t : = M\times\{t\}$ for
$t\in I$ and let $(g_t,K_t)$ be the first and second fundamental form
of $M_t$ in $(L,h)$. The lapse function along $M_t$ is denoted by
$\alpha_t := |\nabL t|^{-1}$. The $K_t$ and $\alpha_t$ are computed
with respect to the future directed unit normal. We will always assume
without further notice that $h$, $g_t$ and $\al_t$ are smooth on all
of $L$ up to the boundary, and furthermore that $\del M \times I$ is
also smooth. For sections~\ref{sec:evol-outerm-mots}
and~\ref{sec:coalescence-mots} we actually need only $C^2$ and in
section~\ref{sec:jump-targ-outerm} and~\ref{sec:regularity} we need
$C^{2,\al}$ regularity.

In this setting a marginally outer trapped tube (MOTT) adapted to
$M_t$ is a smooth three dimensional manifold $\CH$ such that
$\Sigma_t:=\CH \cap M_t$ is a smooth, two-dimensional, embedded MOTS
in $M_t$. Later we will also consider tubes where the $\Sigma_t$
are only immersed.


%
\section{Evolution of outermost MOTSs}
\label{sec:evol-outerm-mots}
In this section we discuss the evolution of the outermost MOTSs in a
Lo\-ren\-tzian spacetime $(L,h)$ as described in
section~\ref{sec:preliminaries}. Assume that $L$ satisfies the null
energy condition (NEC), that is assume that
\begin{equation*}
  \RicL(l,l) \geq 0
\end{equation*}
for all null vectors $l$, where $\RicL$ denotes the Ricci-tensor of
$L$.

We restrict our attention to compact slices with boundary. In
particular, as described in section~\ref{sec:preliminaries}, assume
that $\del M$ has two disconnected parts $\del M = \del^- M \cup
\del^+M$. Note that neither $\del^-M$ nor $\del^+ M$ is assumed to be
connected. We always assume that $\del^+ M$ be non-empty but allow in
certain cases the $\del^- M = \emptyset$. We will subsequently write
$\del M_t$ ($\del^-M_t, \del^+M_t$) to denote $\del M \times \{t\}$
($\del^- M \times \{t\},\del^+ M \times \{t\}$). We assume that
$\theta^+[\del^-M_t] < 0$ whenever $\del^- M_t$ is non-empty and
$\theta^+[\del^+ M_t] >0 $ with respect to the data $(g_t,K_t)$ for
all $t\in[0,T]$. The assumption that $\del^+M$ is non-empty implies
the existence of some outer untrapped surface.

Note that here we do not consider the smooth evolution of MOTSs which
is based on the inverse function theorem and relies on strict
stability as in
\cite{Andersson-Mars-Simon:2005,andersson-mars-simon:2007pub}. The goal
is to formalize a sketch to construct a MOTT given in the previous two
references.

Our main result in this setting is the following.
\begin{theorem}
  \label{thm:mots-evolution}
  Let $L = M \times [0,T]$ be a spacetime satisfying the NEC such that
  $\del M_t = \del ^+ M_t$ with $\theta^+[\del^+ M_t] >0$ or that
  $\del M_t = \del^- M_t \cup \del ^+ M_t$ with $\theta^+[\del^- M_t]
  <0$ and $\theta^+[\del^+ M_t] >0$.

  Assume that the trapped region $\CT_0$ in $(M_0,g_0,K_0)$ is
  non-empty. Then for all $t\in [0,T]$ the trapped region $\CT_t$ of
  $(M_t,g_t,K_t)$ is also non-empty.

  Furthermore, if $J^+(\CT_0)$ denotes the causal future of $\CT_0$ in
  $L$, then we have
  \begin{equation*}
    J^+(\CT_0) \cap M_t \subset \CT_t.
  \end{equation*}
  If this inclusion is not strict at time $\tau>0$, then $J^+(\CT_0)
  \cap M_t = \CT_t$ for all $t\in[0,\tau]$ and $\del^+\CT_t$ satisfies
  $\chi^+ \equiv 0$ and $\RicL(l^+,l^+) = 0$.
\end{theorem}
The interpretation of this theorem is that if there exists an initial
MOTS in $M_0$, then at all later times there also exists a MOTS in
$M_t$ that encloses the points in $M_t$ which are in the causal future
of the trapped region of $M_0$, and thus the terminology \emph{trapped
  region} is indeed justified. If the $\del^+ \CT_t$ form a smooth MOTT,
then this means in particular that this MOTT is achronal.
\begin{proof}
  The proof is based on the Raychaudhuri equation. Denote by
  $\Sigma_\tau:=\del^+ \CT_\tau$ the outermost MOTS in $M_\tau$. Let
  $\Gamma_\tau^+$ denote the null-surface generated by the outgoing
  null normal $l^+$ on $\Sigma_\tau$ and by $\Gamma^+_{\tau,t}$ the
  intersection of $\Gamma_\tau^+$ with $M_t$.

  Since $L$ and all $\al_t$ are smooth, the constants from
  theorem~\ref{thm:estimates} are uniform in $\tau$. This implies that
  there exists a $\delta>0$ depending only on the geometry of $L$ and
  not on the particular $\tau$ such that the surface
  $\Gamma^+_{\tau,t}$ is embedded and homologous to $\del^+ M_t$ for
  all $t\in[\tau,\tau+\delta]$.

  By the Raychaudhuri equation and the null energy condition we know
  that
  \begin{equation*}
    \theta^+(\Gamma^+_{\tau,t}) = -\al_t (|\chi^+|^2 + \RicL(l^+,l^+)) \leq 0
  \end{equation*}
  and thus $\Gamma^+_{\tau,t}$ is contained in the trapped region
  $\CT_t$ as claimed. Note that the inclusion is strict unless
  $|\chi^+|^2 + \RicL(l^+,l^+)= 0$. Since $\Gamma^+_{\tau,t}$ encloses the
  causal future of the region enclosed in $\Sigma_\tau$, we also find
  the inclusion $J^+(\CT_\tau)\subset \CT_t$ for all
  $t\in[\tau,\tau+\delta]$.

  Hence we can start with $\tau=0$ show that the claim holds up to
  time $\delta$ and then restart at time $\delta$ and iterate the
  argument.
\end{proof}
Hence the $\CH_t := \del^+\CT_t$ are non-empty for $t\in[0,T]$ and we
can consider the set
\begin{equation*}
  \CH := \bigcup_{t\in [0,T]} \CH_t \subset M \times [0,T].
\end{equation*}
We will make some remarks about the regularity of this set in
section~\ref{sec:regularity}.


%
\section{Coalescence of MOTSs}
\label{sec:coalescence-mots}
This section is devoted to an informal description of the coalescence
of separate components of the outermost MOTSs into one during the time
evolution. The result here is that if two separate MOTS approach, and
come close enough, then the outermost MOTS must jump before the two
pieces make contact. We work in the same setting as before.

Assume that we have an initial data set with two separate MOTSs
$\Sigma_0^1$ and $\Sigma_0^2$. If these two MOTSs evolve to MOTSs
$\Sigma_t^i$, with $i=1,2$ in such a way that
\begin{equation}
  \dist(\Sigma_t^1, \Sigma_t^2) \to 0
\end{equation}
as $t$ approaches some time $T$, then it has been observed in
numerical simulations that a common MOTS enclosing both $\Sigma_t^1$
and $\Sigma_t^2$ appears before they actually make contact. Here we
show that this has to be the case in general.
\begin{theorem}
  \label{thm:coalescence}
  Let $(M_t,g_t,K_t)$ for all $t\in[0,T]$ be a smooth family of
  initial data sets such that $\del M_t$ splits into disconnected
  parts $\del M_t = \del^- M_t \cup \del^+ M_t$ with
  $\theta^+[\del^-M_t]<0$ and $\theta^+[\del^+M_t]>0$.

  Suppose that for all $t\in[0,T]$ there exists a MOTS $\Sigma_t$
  homologous to $\del^+ M_t$ such that $\Sigma_t$ has at least two
  components $\Sigma_t^1$ and $\Sigma_t^2$ with the property that
  \begin{equation*}
    \dist(\Sigma_t^1, \Sigma_t^2) \to 0 \quad\text{as}\quad t\to T.
  \end{equation*}
  Then there exists a $\tau\in[0,T)$ such that the trapped region
  $\CT_\tau$ of $(M_\tau,g_\tau, K_\tau)$ has one connected component
  which contains both $\Sigma_\tau^1$ and $\Sigma_\tau^2$.
\end{theorem}
\begin{figure}[h!t]
  \centering
  \resizebox{.9\linewidth}{!}{\input{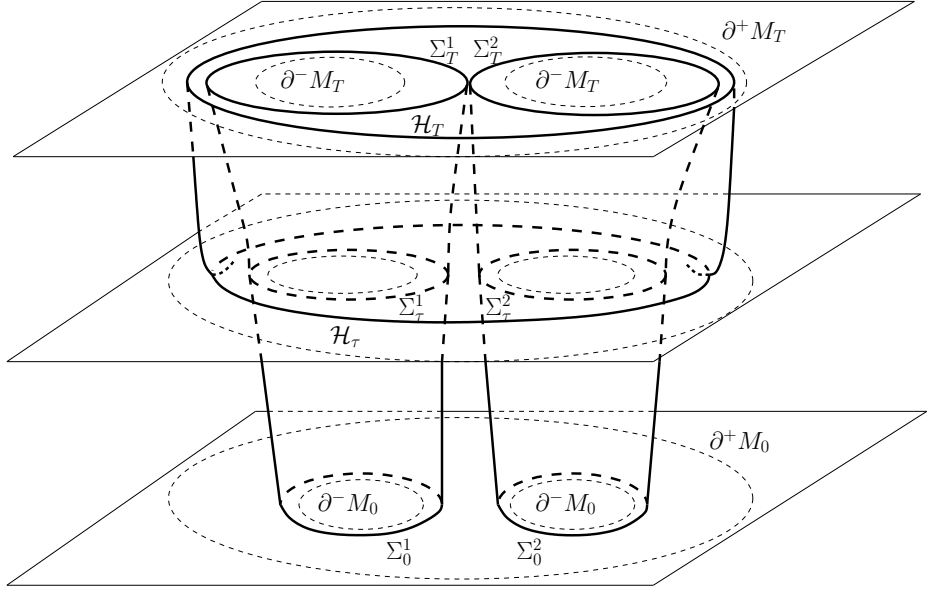}}
  \caption{The situation in Theorem~\ref{thm:coalescence}. Two
    separate MOTS $\Sigma_0^1$ and $\Sigma_0^2$ approach each
    other. At time $\tau<T$ a common enclosing MOTS $\CH_\tau$
    appears, the outermost MOTS thus jumps from two separate pieces to
    a common one. In section~\ref{sec:jump-targ-outerm} we will show
    if such a jump occurs, then the jump target at time $\tau$ will
    bifurcate into two branches at times $>\tau$. This is indicated by
    the dotted lines.}
\end{figure}
\begin{remark}
  The interpretation of the theorem is as follows. Assume that
  initially $\Sigma_t$ is outermost and that the $\Sigma_t$ form a
  smooth MOTT. If $\Sigma_t$ has two components $\Sigma_t^1$ and
  $\Sigma_t^2$ which approach each other, then before they make
  contact, $\Sigma_t$ must stop being outermost, say at time
  $t=\tau$. Hence the outermost MOTS jumps away from
  $\Sigma_\tau$. Instead, $\Sigma^1_\tau\cup \Sigma^2_\tau$ is contained in the
  interior of one connected component of $\CT_\tau$, and thus in the
  trapped region the interiors of $\Sigma^1_\tau$ and $\Sigma^2_\tau$ have
  merged. The outer boundary of this component is therefore a common
  MOTS enclosing $\Sigma_\tau^1$ and $\Sigma_\tau^2$.
\end{remark}
\begin{proof}
  The proof is based on the surgery procedure introduced in
  \cite[Section 6]{Andersson-Metzger:2007}. There, we were able to
  show that a MOTS that comes close to itself can be modified by
  inserting a small neck to construct a weakly outer trapped surface
  outside.

  In \cite[Section 6]{Andersson-Metzger:2007}, some work was invested
  into the point selection for surgery. Here, we do not need the
  special properties, but we do the surgery at the points
  $p_t\in\Sigma_t^1$ and $q_t\in\Sigma_t^2$ which realize the distance
  $\dist(\Sigma_t^1,\Sigma_t^2)$, provided this is small enough. The
  neck to be inserted is then of size comparable to
  $\dist(\Sigma_t^1,\Sigma_t^2)$, and has as axis the geodesic
  $\gamma_t$ joining $p_t$ and $q_t$ in $M$. The rest of the
  construction is otherwise analogous.

  Choose $0\leq \tau<T$ large enough, so that this procedure is
  applicable for all $t \in [\tau,T]$. Denote the region enclosed
  between $\Sigma_t$ and $\del^+ M_t$ by $\Omega_t$ and the geodesic
  joining $\Sigma^1_t$ and $\Sigma^2_t$ by $\gamma_t$. The result of
  the surgery procedure is a weakly outer trapped surface $\Sigma'_t$
  in $\Omega_t\cup \Sigma_t$, enclosing a region $\Omega'_t$ together
  with $\del^+ M$, such that $\Omega_t\setminus \Omega_t'$ contains a
  neighborhood of $\Sigma_t^1 \cup \Sigma_t^2\cup \gamma_t$ in
  $\Omega_t \cup \Sigma_t$.

  An application of theorem \ref{thm:trapped-region} to the manifold
  $\Omega_t'$ with inner boundary $\Sigma_t'$ and outer boundary
  $\del^+ M$ yields an outermost MOTS $\Sigma''_t$ in $\Omega_t'$,
  which is also the outermost MOTS in $M_t$.
\end{proof}


%
\section{Past isolated outermost MOTS}
\label{sec:jump-targ-outerm}
In this section we analyze the question, what happens if the outermost
MOTS jumps in time. To this end assume that $L$ is a spacetime
satisfying the NEC with a foliation 
\begin{equation*}
  L = M \times (-T, T)
\end{equation*}
by spacelike slices $M_t = M \times \{t\}$. As usual, we assume that
$\del M$ is the disjoint union $\del M= \del^- M \cup \del^+M$ and
that with respect to all data sets $(g_t, K_t)$ we have that
$\theta^+[\del^-M] <0$ and $\theta^+[\del^+M] > 0$. Then, in
particular, $M_0$ contains an outermost MOTS $\Sigma$.

We will now assume that $\Sigma\subset M_0$ is the target of a jump
in the outermost MOTSs in the $M_t$ for $t<0$. We formalize this in
the assumption that each component of $\Sigma$ be \emph{stable} and
\emph{past isolated}. Here stability is as defined in
section~\ref{sec:preliminaries}. A MOTS $\Sigma_\tau\subset M_\tau$ is
called \emph{past isolated} if there exists a neighborhood $U$ of
$\Sigma_\tau$ in $L$ such that $M_t \cap U$ does not contain a MOTS
for all $t\in(-T,\tau)$. We say that $\Sigma_\tau$ is \emph{present
  isolated} if there is a neighborhood $V$ of $\Sigma_\tau$ in
$M_\tau$ such that $\Sigma_\tau$ is the only MOTS in $V$.

A jump of the outermost MOTS arises for example in the coalescence of
MOTSs, as described in section~\ref{sec:coalescence-mots}, as after
the jump $\CH_t$ will be past isolated. We show that generically
$\Sigma$ locally splits into two branches of MOTSs in the immediate
future of $\Sigma$. Before we state the actual theorem, we have to
introduce some notation.

Let $n$ be the timelike future unit normal to $M_0$ in $L$ and let
$\nu$ be the spacelike outer unit normal to $\Sigma$ in $M_0$. Then we
define the null frame $l^\pm = n \pm \nu$ along $\Sigma$.

We denote by $W$ the function
\begin{equation*}
  W  = |\chi^+|^2 + \RicL(l^+, l^+)
\end{equation*}
The first term in $W$ is non-negative since it is a sum of squares,
whereas the NEC implies non-negativity of the second term.

We say that $\Sigma$ \emph{satisfies the genericity assumption in the
  spacetime $L$} if
\begin{equation}
  \label{eq:generic}
  W\not\equiv 0
  \quad\text{on}\quad
  \Sigma.
\end{equation}

Denote by $\Gamma^+$ the null-cone generated by the outgoing
null-normal $l^+$ of $\Sigma$ and by $\Gamma^-$ the null-cone
generated by the ingoing null-normal $l^-$. Denote by $\Gamma_t^\pm :
= \Gamma^\pm \cap M_t$ the cross-sections of $\Gamma^\pm$ in
$M_t$. Note that by the above assumptions $\Sigma$ lies in the
interior of $M_0$ and hence so do the $\Gamma_t^\pm$ for $|t|$ small
enough. Since we are only interested in the situation local to
$\Sigma$ we can assume that $T$ is so small that the $\Gamma_t^\pm$
are smooth surfaces in the interior of $M_t$.

In \cite{Andersson-Mars-Simon:2005} a MOTT was constructed near a
strictly stable MOTS. The following argument is an analogue to this
construction if the MOTS $\Sigma$ is only \emph{marginally stable}
that is, when $\Sigma$ is stable but not strictly stable, and
satisfies the genericity condition.
\begin{proposition}
  \label{thm:tube}
  Let $L$ and $(M_t,g_t,K_t)$ be as above. Assume that $\Sigma \subset
  M_0$ is a connected, marginally stable MOTS, that is the principal
  eigenvalue of $L_M$ on $\Sigma$ satisfies $\lambda=0$, and that
  $\Sigma$ satisfies the genericity assumption~\eqref{eq:generic}.

  Then there exists a three dimensional, spacelike, MOTT $\hat\CH$
  containing $\Sigma$ which is tangent to $M_0$ at $\Sigma$. There
  exits a neighborhood $U$ of $\Sigma$ such that all MOTS $\Sigma'
  \subset M_t \cap U$ for $t\in(-T,T)$ are of the form $\hat\CH \cap
  M_t$.
\end{proposition}
\begin{figure}[h!t]
  \centering
  \resizebox{.9\linewidth}{!}{\input{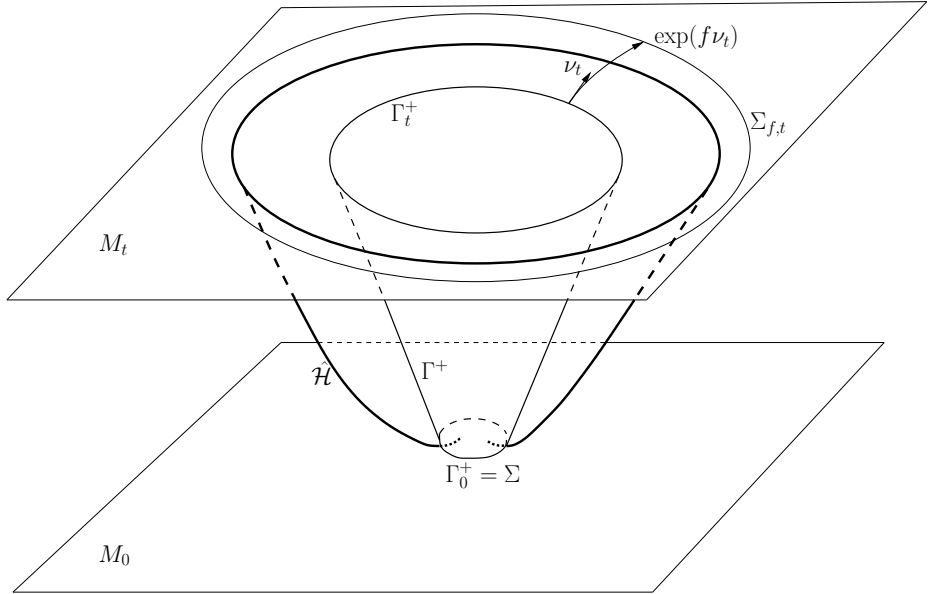}}
  \caption{The situation of Proposition~\ref{thm:tube}. $\Sigma$ is a
    marginally stable MOTS and satisfies the genericity
    assumption. Then it is contained in a MOTT $\hat\CH$ tangent to
    $M_0$.}
\end{figure}
\begin{proof}
  We work in H\"older spaces $C^{2,\alpha}$. Here the choice of
  $\alpha$ is not critical, so we fix one $\alpha>0$ for the remainder
  of the proof.

  We construct the surfaces $\Gamma_t^+$ as above and parametrize
  them by a map
  \begin{equation*}
    G_0 : \Sigma \times (-T,T) \to L 
  \end{equation*}
  such that $\Gamma^+_t = G_0(\Sigma,t)$ and $\dd{G_0}{t} = \alpha_0 l^+$, where
  $\alpha_t$ denotes the lapse function of the foliation $(M_t,g_t,K_t)$.
  
  There exists a neighborhood $\CU\subset C^{2,\alpha}(\Sigma)$ of
  $0\in C^{2,\alpha}(\Sigma)$ such that the map
  \begin{equation*}
    G_f
    :
    \Sigma \times (-T_0,T_0) \to L
    :
    (p,t) \mapsto \exp^{M_t}_{G_0(p,t)}\big(f(p)\nu_t(G_0(p,t))\big),
  \end{equation*}
  is well-defined for all $f\in \CU$. Here $\exp^{M_t}_{G_0(p,t)}$ is
  the exponential map of $M_t$ at $G_0(p,t)\in \Gamma_t^+$ and $\nu_t$
  is the outer normal of $\Gamma_t^+$. We set $\Sigma_{f,t}:=
  G_f(\Sigma,t)$ and assume that $\CU$ and $T$ are small enough to
  ensure that the $\Sigma_{f,t}$ are $C^{2,\alpha}$-surfaces in $M_t$.

  Consider the functional
  \begin{equation*}
    \Theta
    : \CU \times (-T,T) \to C^{0,\alpha}(\Sigma)
    : (f,t) \mapsto G_f(\cdot,t)^* \theta^+[\Sigma_{f,t}].    
  \end{equation*}
  Here $\theta^+[\Sigma_{f,t}]$ denotes $\theta^+$ evaluated on
  $\Sigma_{f,t}$ with respect to the data $(g_t,K_t)$.

  Since
  \begin{equation*}
    \ddeval{G_0}{t}{t=0} = \alpha_0 l^+\quad\text{and}\quad
    \ddeval{G_{sf}(\cdot,0)}{s}{s=0} =f\nu,
  \end{equation*}
  it is well-known (cf.~the setup
  in~\cite{Andersson-Mars-Simon:2005}) that
  \begin{equation*}
    D_f\Theta(0,0)f = L_M f,
    \quad\text{and}\quad
    D_t\Theta(0,0) = - \alpha_0 W
  \end{equation*}
  where $L_M$ is the linearization of $\theta^+$ in $M_0$.

  Since the principal eigenvalue of $L_M$ is equal to zero, the kernel
  of $L_M$ is one-dimensional and spanned by a positive function
  $\phi$. We denote by $X\subset C^{2,\alpha}$ the $L^2$-orthogonal
  complement of $\operatorname{span}\{\phi\}$ in $C^{2,\alpha}$ and
  decompose $C^{2,\alpha}(\Sigma) = X \oplus
  \operatorname{span}\{\phi\}$.

  Denote by $L_M^*$ the (formal) $L^2$-adjoint of $L_M$. Then $L_M^*$
  also has a kernel, which is spanned by a smooth, positive function
  $\psi$. We denote by $Y$ the $L^2$-orthogonal complement of $\operatorname{span}\{\psi\}$ in
  $C^{0,\alpha}(\Sigma)$ and by $P_Y$ the $L_2$-orthogonal projection
  of $C^{0,\alpha}(\Sigma)$ onto $Y$. Then $Y = \operatorname{range}(L_M)$ and
  $L_M|_X: X\to Y$ is an isomorphism.

  In particular, the implicit function theorem \cite[Theorem
  2.7.2]{Nirenberg:1973} implies that for the operator
  \begin{equation*}
    \Theta_Y : \CU \times (-T,T) \to Y : (f,t) \mapsto P_Y(\Theta(f,t))
  \end{equation*}
  there exist constants $\eps>0$, $\delta>0$ and a function
  \begin{equation*}
    x: (-\eps,\eps) \times (-\delta,\delta) \to X
    : (\kappa,t) \mapsto x(\kappa,t)
  \end{equation*}
  with $x(0,0)=0$ such that
  \begin{equation}
    \label{eq:3}
    \Theta_Y(x(\kappa,t) + \kappa\phi, t) = 0.
  \end{equation}
  for all $(\kappa,t) \in (-\eps,\eps) \times (-\delta,\delta)$. The
  uniqueness part of the implicit function theorem furthermore implies
  that all solutions $(y,s) \in X \times (-\delta,\delta)$ to the
  equation $\Theta_Y(y,s)=0$ with $y$ close enough to $0$ are of the
  form $y=x(\kappa,s)$.

  Differentiating equation~\eqref{eq:3} with respect to $\kappa$ at
  $(\kappa,t) = (0,0)$ yields that
  \begin{equation*}
    P_Y L_M \left(\ddeval{x}{\kappa}{(\kappa,t) = (0,0)} + \phi\right) = 0.
  \end{equation*}
  Since $X \cap \operatorname{ker} L_M = \{0\}$ we thus find that
  \begin{equation}
    \label{eq:1}
    \ddeval{x}{\kappa}{(\kappa,t)= (0,0)} = 0.
  \end{equation}    
  To construct MOTSs near $\Sigma$ it thus remains to solve the
  equation
  \begin{equation*}
    (1-P_Y) \Theta\big(x(\kappa,t) + \kappa\phi, t\big) = 0,
  \end{equation*}
  which is a scalar equation in two variables. To solve this equation
  we define the function
  \begin{equation*}
    \vartheta : (-\eps,\eps) \times (-\delta,\delta) \to \IR
    : (\kappa,t) \mapsto \int_\Sigma \psi\, \Theta\big(x(\kappa,t) +
    \kappa\phi, t\big) \dmu.
  \end{equation*}
  Note that by the above 
  \begin{equation*}
    \ddeval{}{t}{(\kappa,t)=(0,0)}\Theta\big(x(\kappa,t) + \kappa\phi, t\big)
    =
    L_M\left(\ddeval{x}{t}{t=0}\right) + D_t \Theta(0,0).
  \end{equation*}
  Since $L_M$ maps into $Y$ which is $L^2$-orthogonal to $\psi$, we
  thus find that
  \begin{equation*}
    D_t \vartheta(0,0) = -\int_\Sigma \alpha_0 W\psi\dmu.
  \end{equation*}
  By the non-degeneracy assumption $W\geq 0$ and $W\not\equiv 0$, we
  find that
  \begin{equation*}
    D_t \vartheta(0,0) < 0. 
  \end{equation*}
  Thus the implicit function theorem implies that there exists a
  $\eps'>0$ and a function $\tau: (-\eps',\eps')\to \IR$ with $\tau(0)=0$ such that
  \begin{equation*}
    \vartheta\big(\kappa, \tau(\kappa)\big) = 0
  \end{equation*}
  for all $\kappa\in (-\eps',\eps')$. Again all solutions close enough
  to zero are of this form.

  As before we can calculate that
  \begin{equation}
    \label{eq:2}
    \ddeval{\tau}{\kappa}{\kappa=0} = - (D_t \vartheta)^{-1}
    (D_\kappa\vartheta)\Big|_{(\kappa,t) = (0,0)} = 0.
  \end{equation}
  Define the map
  \begin{equation*}
    \Phi : \Sigma \times (-\eps',\eps') \to L
    :  (p,\kappa) \mapsto G_{x(\kappa, \tau(\kappa)) + \kappa\phi}\big(p, \tau(\kappa)\big).
  \end{equation*}
  By the above construction the surfaces $\Sigma_\kappa :=
  \Phi(\Sigma, \kappa)$ are MOTS in $M_{\tau(\kappa)}$. In view of
  equations~\eqref{eq:1} and~\eqref{eq:2} we calculate that
  \begin{equation*}
    \ddeval{\Phi}{\kappa}{(\kappa,t)=(0,0)}
    =
    \phi\nu
  \end{equation*}
  This vector field is nowhere zero on $\Sigma$ and normal to
  $\Sigma$. Thus the set
  \begin{equation*}
    \hat\CH := \Phi\big(\Sigma, (-\eps',\eps')\big)     
  \end{equation*}
  is a smooth manifold. Furthermore,
  $\ddeval{\Phi}{\kappa}{(\kappa,t)=(0,0)}$ is spacelike along
  $\Sigma$ and tangent to $M_0$. Hence we infer that $\hat\CH$ is spacelike
  and tangent to $M_0$.
\end{proof}
This proposition implies the main theorem of this section.
\begin{theorem}
  \label{thm:jump}
  Let $L$ and $(M_t,g_t,K_t)$ be as above and assume that $\Sigma
  \subset M_0$ is an outermost MOTS. Let $\Sigma'\subset \Sigma$ be a
  past isolated component of $\Sigma$.  If $\Sigma'$ satisfies the genericity
  assumption~\eqref{eq:generic} there exists a three dimensional,
  spacelike, MOTT $\hat\CH$ containing $\Sigma'$ which lies in $M_0$
  or to
  the future of $M_0$.

  Furthermore if $\Sigma'$ is also present isolated then $\hat\CH$ is such
  that
  \begin{enumerate}
  \item $M_0 \cap \hat\CH = \Sigma'$,
  \item $M_t \cap \hat\CH = \Sigma_t^- \cup \Sigma_t^+$ where
    $\Sigma_t^+$ is a MOTS outside of $\Gamma_t^+$ and $\Sigma^-_t$ is
    a MOTS inside $\Gamma^-_t$, provided $t>0$ is small enough.
  \end{enumerate}
\end{theorem}
\begin{proof}
  Since $\Sigma$ is outermost each of its components and thus
  $\Sigma'$ in particular is stable. If $\Sigma'$ were strictly
  stable, then the result from \cite{Andersson-Mars-Simon:2005} could
  be used to to construct a MOTS extending to the future and past of
  $\Sigma'$, whence $\Sigma'$ can not be strictly stable. Thus
  $\Sigma'$ is marginally stable and proposition~\ref{thm:tube}
  implies the existence of a spacelike MOTT $\hat\CH$ containing
  $\Sigma'$, which is tangent to $\Sigma'$ at $M_0$. Since $\Sigma'$
  is past isolated $\hat\CH$ lies in $M_0$ or to the future of
  $M_0$. Since $\Sigma'$ is part of an outermost MOTS, the outgoing
  part of $\hat\CH$ must lie in the future of $M_0$.

  If $\Sigma'$ is present isolated, then the ingoing part of $\hat\CH$
  also lies to the future of $M_0$. The fact that $\hat\CH$ is tangent
  to $M_0$ and is spacelike implies the claims about the location of
  $\hat\CH$.
\end{proof}
\begin{remark}
  If $\Sigma$ is not present isolated, and the genericity
  assumption~\eqref{eq:generic} holds on $\Sigma$ we can perturb the
  slicing $M_t$ such that $M_0$ changes only inside of $\Sigma$ such
  that $\Sigma$ becomes present isolated.  In general, in the
  perturbed slicing $\Sigma$ need no longer be past isolated, even if
  it was originally.
\end{remark}

In~\cite[Theorem 9.4]{andersson-mars-simon:2007pub} it was shown that if
there is a MOTT $\hat\CH'$ in a slicing $(M_t,g_t,K_t)$ for $t\in [0,T)$
and the MOTS $\Sigma_t' = \hat\CH' \cap M_t$ are connected, strictly
stable and have a smooth limit $\Sigma_T'$ as $t\to T$, then if the
principal eigenvalue of $\Sigma_T'$ is zero, $\hat\CH'$ must be tangent to
$M_T$ along $\Sigma_T'$. Besides the non-degeneracy assumption, the
result there has a further technical assumption, on which we will not
comment.

Proposition~\ref{thm:tube} also implies that $\hat\CH$ is tangent to
$\Sigma_0$ provided its principal eigenvalue is zero. Thus the
argument given here can be used to give another proof of the tangency
property without further technical conditions. We can state the
following version of~\cite[Theorem 9.4]{andersson-mars-simon:2007pub}.
\begin{theorem}
  Let $L \supset M \times [0,T]$ be a partially sliced spacetime
  satisfying the NEC such that the data $(g_t,K_t)$ are smooth on the
  closure of $M \times [0,T]$, the boundary $\del M$ splits into $\del
  M = \del^-M\cup\del^+M$ and $\theta^+[\del^+M_t] >0 $ for all
  $t\in[0,T]$.

  Let $\hat\CH$ be a MOTT adapted to the slicing, such that $\Sigma_t =
  \hat\CH \cap M_t$ is a stable MOTS homologous to $\del^+M$ in $M_t$ for
  $t\in[0,T)$. If the area $|\Sigma_t|$ is bounded as $t\to T$ then
  there exists a stable MOTS $\Sigma_T$ in $M_T$ extending $\hat\CH$.  If
  the principal eigenvalue of $\Sigma_T$ is zero and $\Sigma_T$
  satisfies the genericity assumption~\eqref{eq:generic}, then $\hat\CH$
  is tangent to $\Sigma$.
\end{theorem}
\begin{proof}
  The existence of the limit $\Sigma_T$ follows from the compactness
  of stable MOTS \cite[Theorem 8.1]{Andersson-Metzger:2005}. Although
  convergence there is only asserted in $C^{1,\al}$, in view of
  elliptic regularity for the equation $\theta^+ = 0$ together with
  the $C^{1,\al}$- bounds, this implies $C^k$-convergence for all
  $k\geq 0$, provided $L$ is smooth enough.

  Assuming that $\Sigma_T$ has principal eigenvalue zero and satisfies
  the genericity condition, we can construct a MOTT $\hat\CH'$ near
  $\Sigma_T$ which is tangent to $M_T$ as in the proof of
  theorem~\ref{thm:jump}. Since the implicit function theorem implies
  that $\hat\CH'$ is the unique adapted MOTT near $\Sigma_T$ it has to
  agree with $\hat\CH$ and hence the theorem is proved.
\end{proof}
\begin{remark}
  The previous argument has an interesting implication for the
  continuation of MOTTs. In fact, the constructed $\hat\CH'$ continues
  $\hat\CH$ beyond $\Sigma_T$. However, it is not clear whether this
  continuation does extend $\hat\CH$ to the future. If the $\hat\CH'$ curves
  to the past, we can conclude that the $\Sigma_t$ were not
  outermost. Hence, if the $\Sigma_t$ are outermost for $t\in[0,T)$,
  then $\hat\CH$ can be continued, either as a foliation of $M_T$ by MOTSs
  near $\Sigma_T$ or to the future of $M_T$. In this case also the
  area bound is automatic for the $\Sigma_t$
  (cf.~\cite[Theorem 6.5]{Andersson-Metzger:2007}).
\end{remark}


%
\section{Regularity of MOTTs}
\label{sec:regularity}
In this section, we use the arguments from
section~\ref{sec:jump-targ-outerm} to analyze regularity of the set
\begin{equation*}
  \CH := \bigcup_{t\in [0,T]} \del^+\CT_t \subset M \times [0,T]
\end{equation*}
constructed in section~\ref{sec:evol-outerm-mots}. Before we consider
the more specific setting of section~\ref{sec:jump-targ-outerm} where
the genericity assumption is assumed, we make some general
observations about $\CH$.

Let $\tau \in (0,T)$. The compactness theorem for stable MOTSs in
\cite[Theorem 8.1]{Andersson-Metzger:2005} in combination with the
area estimate \cite[Theorems 6.3 and 6.5]{Andersson-Metzger:2007}
guarantees that as $t\nearrow \tau$ the embedded, stable MOTS $\CH_t$
accumulate on an embedded, stable MOTS $\Sigma^P_{\tau} \subset
\CT_{\tau}$. Using the $C^{1,\alpha}$ result in the reference and
elliptic regularity, we can assume that this is in $C^{2,\alpha}$ if
the ambient spacetime metric is smooth enough, that is $C^{2,\al}$.

We now introduce projections $\pi_{t,\tau} : M_t \to M_\tau$ which
project a point $x\in M_t$ to the intersection of the integral curve
of $\dd{}{t}$ through $x$ with $M_\tau$. By the causal structure of
the $\CH_t$, we find that all projections of the $\CH_t$ for $t<\tau$
lie inside of $\CH_\tau$. As the limit of $\CH_t$ as $t\nearrow \tau$
agrees with the limit of the $\pi_{t,\tau}(\CH_t)$, we see that this
limit is one-sided. Thus we can conclude that the $\CH_t$ actually
converge to a unique limit, which is then given by this
$\Sigma^P_{\tau}$. Since we have a positive lower bound on the outward
injectivity as in theorem \ref{thm:estimates}, we find that the limit
$\Sigma^P_\tau$ must be embedded.

Analogously, we can take a limit of the MOTS $\CH_t$ as
$t\searrow \tau$. Again \cite[Theorem 8.1]{Andersson-Metzger:2005} in
combination with \cite[6.5]{Andersson-Metzger:2007} guarantee that we
get a limit $\Sigma^F_{\tau}$, with convergence from the outside in
the sense given above. However, $\Sigma^F_{\tau}$ need no longer be
embedded since \cite[Theorem 6.3]{Andersson-Metzger:2007} only implies
that the limit can not touch itself from the outside. This causes some
difficulties below.
\begin{definition}
  Assume that $\tau\in(0,T)$.
  \begin{enumerate}
  \item The MOTS $\Sigma^P_{\tau}$ is called \emph{limit from the
      past}, whereas $\Sigma^F_{\tau}$ is called \emph{limit from the
      future}.
  \item If $\Sigma^P_{\tau}\neq \CH_{\tau}$ then $\tau$ is called
    \emph{past jump time}. If $\Sigma^F_{\tau}\neq \CH_{\tau}$ then
    $\tau$ is called \emph{future jump time}.
  \item $\tau$ is called \emph{jump time} if it is either a future or
    past jump time.
  \end{enumerate}
\end{definition}
\begin{remark}
  \begin{enumerate}
  \item By definition $\Sigma^P_{\tau}$ lies in $\CT_{\tau}$ and thus,
    at each jump time $\tau$ the volume between $\Sigma^P_{\tau}$ and
    $\del^+\CT_{\tau}$ is positive. This implies that there are only
    countably many past jump times in $[0,T]$.
  \item Similarly, if $\Sigma^F_{\tau}$ is embedded, then it also lies
    in $\CT_{\tau}$ and thus agrees with $\del^+\CT_{\tau}$. Hence
    $\tau$ is a future jump time, if and only if $\Sigma^F_{\tau}$ is
    not embedded.  In this case the limit is from the outside and thus
    $\Sigma^F_{\tau}$ can not intersect the interior of $\del\CT_\tau$
    since all the projections $\pi_{t,\tau} (\del\CT_t)$ for $t>\tau$
    lie outside of $\CH_\tau$ due to the causal structure of
    $\CH$. Hence there also must be some volume between
    $\CH_{\tau}$ and $\Sigma^F_{\tau}$. This implies that there
    are only countably many future jump times in $[0,T]$.
  \item The causal structure, that is local achronality, of $\CH$
    implies that it is of class $C^{0,1}$ near non-jump times.
  \end{enumerate}  
\end{remark}
This is very little information on the regularity of $\CH$. We
actually expect that the jump times are discrete and that, if the
slicing behaves well, in fact there are only finitely many
jumps. Moreover, what is the regularity of $\CH$ in spacetime at times which 
are not jump times? We will answer these questions below under the genericity 
assumption.

We want to pose a few further interesting questions that we do not
address here. If $\tau$ is a past or future jump time, then one would
also like to compare the area of the $\Sigma^{P/F}_{\tau}$ and
$\CH_{\tau}$.

In case the spacetime settles to a steady state, we expect  that at late 
times $\CH$ is a smooth MOTT and approaches the event horizon, provided the 
spacetime approaches a stationary state. 

Furthermore, in some special situations we expect that $\CH$ is part
of a single smooth MOTT $\hat\CH$ even if there are jump times. An
example of this is described in~\cite{Andersson-Mars-Simon:2005}. It
is then interesting to investigate the causal character of
$\hat\CH\setminus\CH$.

Before we turn to the local regularity theorem, we introduce some
notation. If $I \subset [0,T]$ we denote 
\begin{equation*}
  \CH_I := \bigcup_{t\in I} \CH_{t} \subset \CH.
\end{equation*}
\begin{theorem}
  \label{thm:local_regularity}
  Let $\tau\in (0,T)$ and assume that each component of $\CH_{\tau}$ is
  either strictly stable or satisfies the genericity
  assumption~\eqref{eq:generic}.
  \begin{enumerate}
  \item If $\tau$ is not a past jump time then there exists a $\delta^- =
    \delta^-(\tau)>0$ such that $\CH_{(\tau-\delta^-,\tau]}$ is a smooth
    MOTT.
  \item If $\tau$ is not a future jump time then there exists a
    $\delta^+ = \delta^+(\tau)>0$ such that $\CH_{[\tau,\tau+\delta^+)}$ is a
    smooth MOTT.
  \item If $\tau$ is not a jump time then there exists a
    $\delta=\delta(\tau)>0$ such that $\CH_{(\tau-\delta,\tau+\delta)}$
    is a smooth MOTT. In particular $(\tau-\delta,\tau+\delta)$ does not
    contain further jump-times.
  \end{enumerate}
\end{theorem}
\begin{proof}
  We only show the first assertion, since the second is proved
  similarly and the third is a consequence of the first two.

  Thus assume that $\tau\in (0,T]$ is not a past jump time. In case
  $\CH_{\tau}$ is strictly stable, we can apply the implicit function
  theorem as in \cite{Andersson-Mars-Simon:2005} to construct a smooth
  MOTT extending $\CH_{\tau}$ to the past and the future. In case
  $\CH_{\tau}$ is not strictly stable but satisfies the genericity
  condition~\eqref{eq:generic}, we can apply
  proposition~\ref{thm:tube} to construct an ingoing and outgoing MOTT
  around $\CH_{\tau}$. Since the uniqueness part of the implicit
  function theorem implies that in both cases the respective adapted MOTTs
  are unique near $\CH_{\tau}$, we get that in
  particular the $\CH_t$ for $t\in (\tau-\delta^-, \tau]$ lie on this
  MOTT. Here $\delta(\tau)$ is a positive number depending on the
  geometry of $\CH_{\tau}$ in $L$. Hence $\CH_{(\tau-\delta^-, \tau]}$
  agrees with this MOTT and is smooth.
\end{proof}
The structure of $\CH$ near jump times is analyzed in the following theorem.
\begin{theorem}
  \label{thm:bdry-regularity}
  \begin{enumerate}
  \item Let $\tau\in (0,T)$ be a past jump time, and assume that each
    component of $\Sigma^P_{\tau}$ is either strictly stable or
    satisfies the genericity assumption. Then there exists a
    $\delta^-=\delta^-(\tau)>0$ such that $\CH_{(\tau-\delta^-,\tau)}$ is
    a smooth MOTT which extends to a smooth MOTT $\CH'$ such that
    $\CH'$ includes $\Sigma^P_{\tau}$.
  \item Let $\tau\in (0,T)$ be a future jump time, and assume that each
    component of $\Sigma^F_{\tau}$ is either strictly stable or
    satisfies the genericity assumption. Then there exists a
    $\delta^+=\delta^+(\tau)>0$ such that $\CH_{(\tau, \tau+\delta^+)}$
    is a smooth MOTT which extends to the past as a smooth immersed
    MOTT by adding $\Sigma^F_{\tau}$.
  \end{enumerate}
\end{theorem}
\begin{proof}
  The proof follows from the same argument as before, by an
  application of the implicit function theorem to components of
  $\Sigma^P_{\tau}$ or $\Sigma^F_{\tau}$.

  Note that in particular non-embeddedness of $\Sigma^F_{\tau}$ is not
  an issue, since it may only touch itself from the inside. The
  implicit function theorem in Proposition~\ref{thm:tube} can also be
  applied to immersed surfaces to construct an immersed tube around
  $\Sigma^F_\tau$. As this construction implies that the scalar
  product of the future pointing or outward tangent to the tube and
  the outer normal to $\Sigma^F_\tau$ is positive, we infer that the
  MOTS along the tube which are outside or to the future of
  $\Sigma^F_\tau$ are indeed embedded.
\end{proof}
Combining theorems~\ref{thm:local_regularity}
and~\ref{thm:bdry-regularity} we arrive at the following global
statement.
\begin{theorem}
  Assume that all components of the following MOTS are either strictly
  stable or satisfy the genericity assumption~\eqref{eq:generic}:
  \begin{enumerate}
  \item $\CH_t$, for all $t\in[0,T]$,
  \item $\Sigma^P_t$ whenever $t$ is a past jump time, and
  \item $\Sigma^F_t$ whenever $t$ is a future jump time.
  \end{enumerate}
  Then there are finitely many times
  \begin{equation*}
    0 = \tau_0 < \tau_1 < \ldots \tau_N < \tau_{N+1} = T
  \end{equation*}
  such that each $\tau_k$ for $k=1,\ldots,N$ is a jump time and the piece
  \begin{equation*}
    \CH_{(\tau_k,\tau_{k+1})}
  \end{equation*}
  for $k=0,\ldots,N$ is a smooth MOTT, which can be extended as a
  smooth immersed MOTT by adding $\Sigma^F_{\tau_k}$ in the past, and
  $\Sigma^P_{\tau_{k+1}}$ in the future.
\end{theorem}


%
\section{The relation of outer trapped surfaces to singularities}
\label{sec:relat-outer-trapp}
In this section, we revisit the classical singularity theorem of
Penrose \cite[Section 8]{Hawking-Ellis:1973} in the perspective of
outer trapped surfaces. In particular, we want to clarify that under
suitable assumptions, the presence of an outer trapped surface ---
without assumptions on the ingoing expansion --- implies that a
spacetime is not geodesically complete. Although the results in this
section are not new, they do not seem to be well-known
either. Therefore we give a short presentation of the arguments
involved. The ideas presented here appear in the classical reference
\cite{Hawking-Ellis:1973} and in \cite{gannon:1976}. In fact, our
argument is very close to \cite{totschnig:1994}.
\begin{theorem}
  Let $(L,h)$ be a globally hyperbolic Lorentzian spacetime
  satisfying the null energy condition $\RicL(v,v)\geq 0$ for all
  null-vectors $v$.

  Assume that $L$ contains a Cauchy surface $M$ such that $\Sigma
  \subset M$ is a $C^2$-surface which separates $M$ into two
  disconnected parts $M \setminus \Sigma = M^- \cup M^+$. Let the
  outer normal along $\Sigma$ be the one pointing into $M^+$. If
  $\theta^+[\Sigma] <0$, where $\theta^+[\Sigma]$ is calculated with
  respect to this choice of normal, and $M^+$ is a connected,
  non-compact manifold with boundary $\Sigma$, then $L$ is not null
  geodesically complete.
\end{theorem}
\begin{remark}
  Assume that $M$ and $\Sigma$ are as above, except that $\Sigma$ is a
  stable MOTS instead of having $\theta^+<0$. Then, if each component
  of $\Sigma$ is either strictly stable or satisfies the genericity
  assumption~\eqref{eq:generic} then $M$ and $\Sigma$ can locally
  be deformed in $L$ to a Cauchy surface $M'$ and a surface $\Sigma'$
  with $\theta^+(\Sigma) < 0$.
\end{remark}
\begin{proof}
  As usual we will assume that $L$ is geodesically complete and deduce
  a contradiction. We denote by $J^+(\Sigma)$ the future causal
  development of $\Sigma$ in $L$. Its boundary $\del J^+(\Sigma)$ is
  generated by null-geodesic segments with past endpoints on $\Sigma$
  and orthogonal to $\Sigma$. Denoting by $l^+$ and $l^-$ a choice of
  outgoing and ingoing null normal fields. Then the generators of
  $\del J^+(\Sigma)$ are tangent to either $l^+$ or $l^-$ where they
  meet $\Sigma$.

  Assume that $p\in \del J^+(\Sigma)$ can be connected to $\Sigma$ by
  a null geodesic $\gamma_1:[0,1]\to L$ such that $\gamma_1(0) \in
  \Sigma$, $\dot\gamma_1(0)=l^+$ and $\gamma_1(1)=p$. Then it can not
  happen that there is also a null geodesic $\gamma_2:[0,1]\to L$ such
  that $\gamma_2(0) \in \Sigma$, $\dot\gamma_2(0)=l^-$ and
  $\gamma_2(1)=p$. This can be seen as follows. First note that
  $\gamma_i(t)\in\del J^+(\Sigma)$ for all $t\in(0,1]$. One can
  define a continuous curve $\gamma:[0,2]\to \del J^+(\Sigma)$ such
  that
  \begin{equation*}
    \gamma(t) =
    \begin{cases}
      \gamma_1(t) & t\in[0,1]\\
      \gamma_2(2-t) & t\in[1,2].
    \end{cases}
  \end{equation*}
  Let $\tau$ be a time function on $L$. Then we can define the
  projection $\Phi:L\to M$ such that $\Phi(p) = q$ if and only if $p$
  lies on the integral curve of $\nabla \tau$ which meets $M$ at the
  point $q$. Since $L$ is globally hyperbolic, $\Phi$ is well defined
  on all of $L$. Define $\tilde\gamma:[0,2]\to M$ as $\tilde\gamma(t)
  = \Phi(\gamma(t))$ for all $t\in[0,2]$. This continuous curve starts
  and ends on $\Sigma$ and has the property that $\tilde\gamma(t)\in
  M^+$ for $t\in(0,\eps)$ and $\tilde\gamma(t)\in M^-$ for
  $t\in(2-\eps,2)$. By continuity, and since $\Sigma$ separates, there
  exists $t_0\in[\eps, 2-\eps]$ with $\tilde\gamma(t_0)\in\Sigma$. By
  definition, this means that $\Phi(\gamma(t_0))\in\Sigma$, but this
  is impossible, as it would imply that $\gamma(t_0)\in\del
  J^+(\Sigma)\cap I^+(\Sigma)$, where $I^+(\Sigma)$ denotes the
  chronological future of $\Sigma$.

  We thus infer that $\del J^+(\Sigma) \setminus \Sigma$ splits into
  two parts $\del J^+(\Sigma) = H^+ \cup H^-$ where $H^+$ is generated
  by the outgoing null-geodesic segments and $H^-$ is generated by
  ingoing null-geodesic segments. The standard convergence results for
  geodesic congruences imply that each null-geodesic leaves $H^+$
  after a finite value of the affine parameter. This implies that
  $H^+$ is a compact Lipschitz manifold with boundary $\Sigma$.

  By the above argument it is easy to see that $\Phi$ maps $H^+$ into
  $M^+$. Since $M^+$ is non-compact but $\Phi(H^+)$ is, we have that
  $M^+\neq\Phi(H^+)$. Then $\Phi(H^+)$ must have a boundary besides
  $\Sigma$ in $M^+$ as $M^+$ is connected. This is not possible as
  this would imply that $\nabla\tau$ is tangent to $H^+$
  somewhere. This yields the desired contradiction.
\end{proof}


%
\section*{Acknowledgments}
The authors would like to thank Jeffrey Winicour and Bela Szilagyi for
interesting discussions, in particular in relation with
section~\ref{sec:coalescence-mots}.

LA, MM and JM thank the Mittag-Leffler-Institute, Djursholm, Sweden
for hospitality and support. LA is supported in part by the NSF, under
contract no.\ DMS 0407732 and DMS 0707306 with the University of
Miami. MM and WS were supported by the Spanish MEC project
FIS2006-05319. MM was also supported by the project P06-FQM-01951 of
the J. de Analuc\'{\i}a. We also thank the anonymous referees for
their helpful comments.


%
\bibliographystyle{amsalpha}
\bibliography{../extern/references}
\end{document}